\documentclass[11pt]{article}

\usepackage[margin=1in]{geometry}
\usepackage{amsmath, amssymb, amsthm}
\usepackage{booktabs}
\usepackage{graphicx}
\usepackage{multirow}
\usepackage{array}
\usepackage{natbib}
\usepackage{url}
\usepackage{authblk}
\usepackage{setspace}
\usepackage{xcolor}

\usepackage[
    colorlinks=true,
    linkcolor=blue,
    citecolor=blue,
    urlcolor=blue
]{hyperref}

\newtheorem{theorem}{Theorem}

\begin{document}

\title{A Continuous-Weighted Win Ratio for Hierarchical Composite Endpoints}

\author[1]{Kexuan Li\thanks{
Email:
\href{mailto:kexuan.li.77@gmail.com}{kexuan.li.77@gmail.com}.

This manuscript is a working paper and comments are welcome. The author apologizes for any remaining rough edges; this draft was prepared in limited spare time between other duties.
}}

\date{\today}

\maketitle

\begin{abstract}
Hierarchical composite endpoints are commonly used in clinical trials when component outcomes differ in clinical importance. 
The win ratio compares patients across treatment groups according to a pre-specified order of clinical priority and has an intuitive interpretation. 
A strict hierarchical rule, however, uses a lower-priority endpoint only when all higher-priority endpoints are tied or non-informative. 
This may reduce power when the treatment effect is mainly expressed through lower-priority outcomes. 
Recent threshold-based extensions relax this hierarchy by allowing lower-priority outcomes to contribute when higher-priority outcomes are within pre-specified margins. We propose a continuous-weighted win ratio for hierarchical composite endpoints. 
The method replaces hard transitions between endpoint levels with smooth weights, so that a lower-priority endpoint can contribute gradually when the higher-priority endpoint is close to tied. 
The resulting estimator is a two-sample U-statistic. 
We derive its large-sample distribution using the Hoeffding projection and provide a consistent plug-in variance estimator. 
We also describe the role of the tuning parameter through local alternatives and local efficiency. 
The proposed class includes the strict win ratio and hard-threshold rules as special or limiting cases. 
Simulation studies show that relaxing the strict hierarchy can improve power when treatment effects are concentrated on lower-priority endpoints, but may reduce power when the highest-priority endpoint carries the main treatment signal.
\end{abstract}


\section{Introduction}

Clinical trials often evaluate treatment benefit using multiple outcomes. 
These outcomes may differ substantially in clinical importance. 
For example, in cardiovascular trials, death is typically more important than hospitalization, and hospitalization may be more important than a continuous symptom or functional measure. 
A conventional composite endpoint may be difficult to interpret in this setting, because clinically severe and less severe events can contribute similarly to the analysis. 
This concern has motivated methods that compare patients according to a pre-specified clinical hierarchy rather than treating all components as exchangeable.

The win ratio of \citep{pocock2012win} is one of the most widely used approaches for prioritized composite endpoints. 
For each treatment-control pair, outcomes are compared sequentially according to clinical priority until the pair is classified as a treatment win, a control win, or a tie. 
The method is closely related to the Finkelstein--Schoenfeld test \citep{finkelstein1999combining} and to generalized pairwise comparisons \citep{buyse2010generalized}. 
These approaches use pairwise comparisons across treatment groups and can accommodate outcomes of different types, including time-to-event, binary, ordinal, and continuous endpoints. 
They also lead to several related treatment-effect summaries, including the win ratio, win odds, and net benefit \citep{dong2020win, brunner2021win}.

The statistical properties of win-ratio-type procedures have been studied from several perspectives. 
\citep{bebu2016large} developed large-sample inference for win-ratio analyses of composite outcomes with prioritized components. 
\citep{oakes2016win} studied the win-ratio statistic for trials with multiple event types. 
For censored outcomes, extensions of generalized pairwise comparisons have been proposed to reduce bias and improve efficiency \citep{peron2018extension}. 
Regression formulations have also been developed, including proportional win-fractions models for composite outcomes \citep{mao2021class}. 
Together, these developments have made pairwise comparison methods a useful class of tools for analyzing hierarchical endpoints.

A strict hierarchical rule has an important limitation. 
Lower-priority endpoints are used only when higher-priority endpoints do not determine the pairwise comparison. 
Consequently, if the treatment effect is weak or absent on the highest-priority endpoint but stronger on a lower-priority endpoint, the standard win ratio may have low power. 
This situation is not rare in practice. 
A treatment may have limited impact on mortality over the follow-up period but may improve hospitalization, symptoms, or other clinically meaningful outcomes. 
Recent work has therefore considered ways to relax the strict hierarchy. 
For example, \citep{mou2024win} proposed win-ratio methods with multiple thresholds, allowing lower-priority outcomes to contribute when higher-priority outcomes are within specified margins.

In this article, we propose a continuous-weighted version of the win ratio. 
Instead of using a hard threshold to decide whether the comparison moves from one endpoint level to the next, we use a smooth weight that depends on how close the higher-priority endpoint is to a tie. 
The resulting pairwise score is a convex combination of endpoint-level scores. 
This construction keeps the score centered under the null hypothesis of no treatment difference and preserves the two-sample U-statistic structure. 
The proposed method is not intended to uniformly dominate the standard win ratio. 
Rather, it provides a way to study and control the trade-off between preserving the priority of higher-ranking endpoints and recovering information from lower-priority endpoints.

The remainder of the article is organized as follows. 
Section 2 introduces notation and reviews hierarchical pairwise comparisons. 
Section 3 defines the continuous-weighted win-ratio score. 
Section 4 gives the estimand, large-sample inference, and variance estimation. 
Section 5 discusses the tuning parameter through local alternatives and local efficiency. 
Section 6 presents simulation studies, and Section 7 concludes with practical considerations.

\section{Hierarchical pairwise comparison}

Consider a two-arm randomized trial with \(n_1\) patients in the treatment group and \(n_0\) patients in the control group. 
Let \(Y_i=(Y_{i1},\ldots,Y_{iK})\), \(i=1,\ldots,n_1\), denote the prioritized outcomes for treatment patients, and \(X_j=(X_{j1},\ldots,X_{jK})\), \(j=1,\ldots,n_0\), denote the corresponding outcomes for control patients. 
The endpoints are ordered from highest to lowest clinical priority. 
For notational simplicity, all endpoints are coded so that larger values are favorable. 
For a treatment-control pair, write \(d_{ijk}=Y_{ik}-X_{jk}\).

The standard hierarchical win-ratio rule compares \(d_{ij1}\) first. 
If the first endpoint determines a winner, the pairwise comparison stops. 
Only if the first endpoint is tied or non-informative does the comparison proceed to the second endpoint, and so on. 
Let \(S_{\rm WR}(Y_i,X_j)\in\{-1,0,1\}\) denote the resulting score, with 1 indicating a treatment win, \(-1\) a control win, and 0 an unresolved tie. 
The corresponding net-benefit estimand is \(\Delta_{\rm WR}=E\{S_{\rm WR}(Y,X)\}\), estimated by
\[
\widehat{\Delta}_{\rm WR}
=
(n_1n_0)^{-1}
\sum_{i=1}^{n_1}\sum_{j=1}^{n_0}S_{\rm WR}(Y_i,X_j).
\]
The usual win ratio is the ratio of the probability of a treatment win to the probability of a control win. 
In this article we work mainly on the net-benefit scale, because it is centered at zero under the null and leads to a standard two-sample U-statistic. 
A ratio-scale summary can be reported as a secondary measure when desired.

\section{Continuous-weighted win-ratio score}

For each endpoint \(k\), let \(q_k(d)\in[-1,1]\) be a pairwise component score. 
Positive values favor treatment, negative values favor control, and zero indicates no preference. 
A hard component score is \(q_k(d)=1\) if \(d>c_k\), \(q_k(d)=-1\) if \(d<-c_k\), and \(q_k(d)=0\) otherwise, where \(c_k\ge 0\) is a clinical relevance margin. 
A smooth component score may also be used, for example \(q_k(d)=2\{1+\exp(-d/\beta_k)\}^{-1}-1\). 
Throughout, we assume \(q_k(-d)=-q_k(d)\), so that reversing treatment and control reverses the component score.

The proposed method modifies the transition between endpoint levels. 
Let \(w_k(|d|;\tau_k)\in[0,1]\) be a non-increasing function of \(|d|\). 
When \(|d|\) is small, the higher-priority endpoint is close to tied and \(w_k\) is close to 1, allowing the next endpoint to contribute more. 
When \(|d|\) is large, \(w_k\) is close to 0 and the pairwise comparison is mostly determined by the current endpoint. 
Examples include the Gaussian weight
\[
w_k(|d|;\tau_k)=\exp\{-d^2/(2\tau_k^2)\}
\]
and the logistic weight
\[
w_k(|d|;\tau_k,a_k)=\{1+\exp[(|d|-\tau_k)/a_k]\}^{-1}.
\]

For two endpoints, the continuous-weighted score is
\[
S_\eta(Y,X)
=
\{1-w_1(|d_1|;\tau_1)\}q_1(d_1)
+
w_1(|d_1|;\tau_1)q_2(d_2),
\]
where \(d_k=Y_k-X_k\), and \(\eta\) denotes the collection of weights, margins, and tuning parameters. 
The convex-combination form is important. 
A purely additive score, such as \(q_1(d_1)+w_1(|d_1|)q_2(d_2)\), is generally not centered under the null and can create a systematic shift in favor of one group. 
In contrast, the convex form assigns total weight one across endpoint levels.

For \(K\) endpoints, define \(A_1=1\) and \(A_{k+1}=A_kw_k(|d_k|;\tau_k)\) for \(k=1,\ldots,K-1\), with \(w_K\equiv 0\). 
The general score is
\[
S_\eta(Y,X)
=
\sum_{k=1}^{K}
A_k\{1-w_k(|d_k|;\tau_k)\}q_k(d_k).
\]
Because \(A_k\ge 0\) and the weights over the \(K\) endpoint levels sum to one, \(S_\eta(Y,X)\in[-1,1]\). 
Moreover, if each \(w_k\) depends on \(d_k\) only through \(|d_k|\) and each \(q_k\) is antisymmetric, then
\[
S_\eta(Y,X)=-S_\eta(X,Y).
\]
Therefore, under the global null in which \(Y\) and \(X\) have the same distribution, \(E\{S_\eta(Y,X)\}=0\).

The standard win-ratio rule and hard-threshold rules are included as special cases. 
If \(w_k(|d|)=I(d=0)\) and \(q_k(d)=\operatorname{sign}(d)\), the strict lexicographic rule is recovered. 
If \(w_k(|d|)=I(|d|\le \tau_k)\), the rule becomes a threshold-based hierarchical comparison. 
The continuous-weighted version replaces this discontinuous transition with a smooth one, while retaining the same pairwise comparison structure.

\section{Estimand and large-sample inference}

For a fixed choice of \(\eta\), define the treatment-effect estimand
\[
\Delta_\eta=E\{S_\eta(Y,X)\}.
\]
The corresponding estimator is
\[
\widehat{\Delta}_\eta=(n_1n_0)^{-1}\sum_{i=1}^{n_1}\sum_{j=1}^{n_0}S_\eta(Y_i,X_j).
\]
The primary null hypothesis is \(H_0:\Delta_\eta=0\), with one-sided alternative \(H_1:\Delta_\eta>0\). 
The tuning parameter \(\eta\) is treated as fixed and pre-specified in the inferential results below.

Let \(N=n_1+n_0\) and assume \(n_1/N\to\rho\in(0,1)\). 
Write \(h_\eta(y,x)=S_\eta(y,x)\). 
Define
\[
\phi_1(y)=E\{h_\eta(y,X)\}-\Delta_\eta,\qquad
\phi_0(x)=E\{h_\eta(Y,x)\}-\Delta_\eta.
\]
The asymptotic variance is
\[
\sigma_\eta^2=\rho^{-1}\operatorname{Var}\{\phi_1(Y)\}
+(1-\rho)^{-1}\operatorname{Var}\{\phi_0(X)\}.
\]

\begin{theorem}
Assume \(h_\eta(Y,X)=S_\eta(Y,X)\) is bounded, 
\(\sigma_\eta^2>0\), and \(n_1/N\to\rho\in(0,1)\), where \(N=n_1+n_0\). 
Then
\[
\sqrt{N}(\widehat{\Delta}_\eta-\Delta_\eta)
\overset{d}{\longrightarrow}
N(0,\sigma_\eta^2),
\]
where
\[
\sigma_\eta^2
=
\rho^{-1}\operatorname{Var}\{\phi_1(Y)\}
+
(1-\rho)^{-1}\operatorname{Var}\{\phi_0(X)\}.
\]
\end{theorem}

\begin{proof}
Let \(h(y,x)=h_\eta(y,x)\) for notational simplicity, and write
\[
\Delta=E\{h(Y,X)\}.
\]
The estimator is
\[
\widehat{\Delta}
=
\frac{1}{n_1n_0}
\sum_{i=1}^{n_1}
\sum_{j=1}^{n_0}
h(Y_i,X_j).
\]
Define the first-order projection terms
\[
\phi_1(y)=E\{h(y,X)\}-\Delta,
\qquad
\phi_0(x)=E\{h(Y,x)\}-\Delta.
\]
By construction, \(E\{\phi_1(Y)\}=0\) and \(E\{\phi_0(X)\}=0\).

Now define the second-order remainder kernel
\[
r(y,x)
=
h(y,x)-\Delta-\phi_1(y)-\phi_0(x).
\]
This kernel is degenerate in both arguments. Indeed,
\[
E\{r(y,X)\}
=
E\{h(y,X)\}-\Delta-\phi_1(y)-E\{\phi_0(X)\}
=
0,
\]
and similarly,
\[
E\{r(Y,x)\}=0.
\]
Using this decomposition,
\[
h(Y_i,X_j)-\Delta
=
\phi_1(Y_i)+\phi_0(X_j)+r(Y_i,X_j).
\]
Averaging over all treatment-control pairs gives the Hoeffding decomposition
\[
\widehat{\Delta}-\Delta
=
\frac{1}{n_1}\sum_{i=1}^{n_1}\phi_1(Y_i)
+
\frac{1}{n_0}\sum_{j=1}^{n_0}\phi_0(X_j)
+
R_{n_1,n_0},
\]
where
\[
R_{n_1,n_0}
=
\frac{1}{n_1n_0}
\sum_{i=1}^{n_1}
\sum_{j=1}^{n_0}
r(Y_i,X_j).
\]

We next show that the remainder is negligible at the \(\sqrt{N}\) scale. 
Because \(h\) is bounded, \(r\) has finite second moment. 
Furthermore, the degeneracy of \(r\) implies that most covariance terms vanish. 
Specifically,
\[
E\{r(Y_i,X_j)r(Y_{i'},X_{j'})\}=0
\]
whenever \(i\ne i'\) or \(j\ne j'\), except for the case in which both indices coincide. 
For example, if \(i=i'\) but \(j\ne j'\), then conditional on \(Y_i\),
\[
E\{r(Y_i,X_j)r(Y_i,X_{j'})\mid Y_i\}
=
E\{r(Y_i,X_j)\mid Y_i\}
E\{r(Y_i,X_{j'})\mid Y_i\}
=
0.
\]
The same argument applies when \(j=j'\) but \(i\ne i'\). 
Therefore
\[
\operatorname{Var}(R_{n_1,n_0})
=
\frac{1}{n_1^2n_0^2}
\sum_{i=1}^{n_1}\sum_{j=1}^{n_0}
E\{r(Y_i,X_j)^2\}
=
O\{(n_1n_0)^{-1}\}.
\]
Since \(n_1/N\to\rho\in(0,1)\), both \(n_1\) and \(n_0\) are of order \(N\), so
\[
N\operatorname{Var}(R_{n_1,n_0})
=
O(N^{-1})\to 0.
\]
Thus \(\sqrt{N}R_{n_1,n_0}=o_p(1)\).

It remains to study the two first-order projection sums. 
The treatment and control samples are independent, so the two sums are independent. 
By the central limit theorem,
\[
\sqrt{N}
\frac{1}{n_1}
\sum_{i=1}^{n_1}\phi_1(Y_i)
=
\sqrt{\frac{N}{n_1}}
\left\{
\frac{1}{\sqrt{n_1}}
\sum_{i=1}^{n_1}\phi_1(Y_i)
\right\}
\overset{d}{\longrightarrow}
N\left(0,\rho^{-1}\operatorname{Var}\{\phi_1(Y)\}\right),
\]
and similarly
\[
\sqrt{N}
\frac{1}{n_0}
\sum_{j=1}^{n_0}\phi_0(X_j)
\overset{d}{\longrightarrow}
N\left(0,(1-\rho)^{-1}\operatorname{Var}\{\phi_0(X)\}\right).
\]
Because the two limiting normal variables are independent, their sum is normal with variance
\[
\sigma_\eta^2
=
\rho^{-1}\operatorname{Var}\{\phi_1(Y)\}
+
(1-\rho)^{-1}\operatorname{Var}\{\phi_0(X)\}.
\]
Combining this with \(\sqrt{N}R_{n_1,n_0}=o_p(1)\) proves the result.
\end{proof}

A plug-in variance estimator can be obtained from empirical projections. 
Let
\[
\widehat{\phi}_{1i}=n_0^{-1}\sum_{j=1}^{n_0}h_\eta(Y_i,X_j)-\widehat{\Delta}_\eta,\qquad
\widehat{\phi}_{0j}=n_1^{-1}\sum_{i=1}^{n_1}h_\eta(Y_i,X_j)-\widehat{\Delta}_\eta.
\]
Let \(s_1^2\) and \(s_0^2\) be the sample variances of \(\widehat{\phi}_{1i}\) and \(\widehat{\phi}_{0j}\), respectively. 
With \(\widehat{\rho}=n_1/N\), define
\[
\widehat{\sigma}_\eta^2=\widehat{\rho}^{-1}s_1^2+(1-\widehat{\rho})^{-1}s_0^2.
\]

The empirical projection variance estimator is consistent under the same conditions. 
To see this, note that
\[
\widehat{\phi}_{1i}
=
\frac{1}{n_0}\sum_{j=1}^{n_0}h(Y_i,X_j)-\widehat{\Delta}
\]
is an empirical version of \(\phi_1(Y_i)=E\{h(Y_i,X)\}-\Delta\). 
For fixed \(Y_i\), the average \(n_0^{-1}\sum_j h(Y_i,X_j)\) converges in probability to \(E\{h(Y_i,X)\}\), and \(\widehat{\Delta}\to_p\Delta\). 
Because the kernel is bounded, the convergence is sufficiently regular to imply that the sample variance of \(\widehat{\phi}_{1i}\) converges in probability to \(\operatorname{Var}\{\phi_1(Y)\}\). 
The same argument applies to the control projection. 
Therefore, with \(\widehat{\rho}=n_1/N\),
\[
\widehat{\sigma}_\eta^2
=
\widehat{\rho}^{-1}s_1^2
+
(1-\widehat{\rho})^{-1}s_0^2
\to_p
\sigma_\eta^2.
\]
The Wald statistic
\[
Z_\eta=\frac{\sqrt{N}\widehat{\Delta}_\eta}{\widehat{\sigma}_\eta}
\]
is asymptotically standard normal under \(H_0\), and a one-sided level-\(\alpha\) test rejects when \(Z_\eta>z_{1-\alpha}\).
A ratio-scale summary may also be reported. 
One option is \(\Omega_\eta=(1+\Delta_\eta)/(1-\Delta_\eta)\), with estimator \(\widehat{\Omega}_\eta=(1+\widehat{\Delta}_\eta)/(1-\widehat{\Delta}_\eta)\). 
For scores taking values in \(\{-1,0,1\}\), this is a win-odds type summary. 
For continuous pairwise scores, we use \(\Delta_\eta\) as the primary estimand because it has a direct U-statistic form and a simple null value.

\section{Role of the tuning parameter}

The tuning parameter determines how rapidly the pairwise comparison moves from a higher-priority endpoint to the next endpoint. 
Small values produce a rule close to the strict hierarchy, whereas large values allow lower-priority endpoints to contribute more often. 
This is a clinically meaningful choice. 
If the treatment effect is concentrated on the highest-priority endpoint, excessive weighting of lower-priority endpoints may reduce power. 
If the treatment effect is mainly on a lower-priority endpoint, the strict hierarchy may have little power.

The trade-off can be described through local alternatives. 
Consider a regular submodel indexed by a vector \(b=(b_1,\ldots,b_K)\), where \(b_k\) describes the local effect direction for endpoint \(k\). 
Suppose that, under a sequence of local alternatives, 
\[
\Delta_{\eta,N}(b)=E_b\{S_\eta(Y,X)\}=N^{-1/2}\dot{\Delta}_\eta(b)+o(N^{-1/2}).
\]
Then the statistic \(Z_\eta\) has limiting distribution
\[
Z_\eta \overset{d}{\longrightarrow} N\{\dot{\Delta}_\eta(b)/\sigma_\eta,1\}.
\]
The local power is therefore
\[
1-\Phi\{z_{1-\alpha}-\dot{\Delta}_\eta(b)/\sigma_\eta\},
\]
and the local efficacy is proportional to
\[
\mathcal{E}_\eta(b)=\dot{\Delta}_\eta(b)^2/\sigma_\eta^2.
\]
Thus, for a fixed clinically plausible direction \(b\), one may choose \(\eta\) to maximize \(\mathcal{E}_\eta(b)\). 
If the direction of treatment effect is uncertain, an average criterion \(E_\Pi\{\mathcal{E}_\eta(b)\}\) or a maximin criterion \(\min_{b\in\mathcal{B}}\mathcal{E}_\eta(b)\) may be considered over a pre-specified set of alternatives. 
Such choices should be made before unblinding in confirmatory trials; otherwise, the adaptive choice of \(\eta\) would need to be accounted for in inference.

The two-endpoint case gives useful intuition. 
Let
\[
S_\tau(d_1,d_2)=\{1-w_\tau(d_1)\}q_1(d_1)+w_\tau(d_1)q_2(d_2),
\]
where \(w_\tau(d_1)\) is an even function of \(d_1\). 
Under a smooth location-shift model with local shift \(b=(b_1,b_2)\), and assuming differentiability of \(q_1,q_2\) and \(w_\tau\), the first-order signal can be written as
\[
\dot{\Delta}_\tau(b)
=
b_1 E\left[
\{1-w_\tau(D_1)\}q_1'(D_1)
+
w_\tau'(D_1)\{q_2(D_2)-q_1(D_1)\}
\right]
+
b_2 E\{w_\tau(D_1)q_2'(D_2)\},
\]
where \(D_k=Y_k-X_k\) under the null. 
This expression shows how the tuning parameter affects the two sources of signal. 
The term involving \(b_1\) represents information from the higher-priority endpoint, modified by the transition weight. 
The term involving \(b_2\) is multiplied by \(w_\tau(D_1)\), so lower-priority information contributes more when the weight is larger. 
The variance \(\sigma_\tau^2\) is determined by the same U-statistic projection variance described in Section 4. 
Therefore the optimal value of \(\tau\) depends on both the anticipated endpoint-level treatment effects and the variance induced by the pairwise score.

For non-smooth component scores, such as \(q_k(d)=\operatorname{sign}(d)\), the same interpretation can be obtained by replacing the derivative calculation with finite differences or by using a smoothed approximation to the sign function. 
This is sufficient for selecting a pre-specified value of \(\tau\) in design-stage evaluations and simulations.

\section{Simulation study 1: two continuous endpoints}

We first evaluated the proposed approach in a simple setting with two continuous endpoints. 
This simulation was intended to study the behavior of the weighting rule in a controlled setting before considering time-to-event outcomes and censoring. 
The first endpoint was treated as the higher-priority endpoint and the second endpoint as the lower-priority endpoint. 
Both endpoints were coded such that larger values represented better outcomes.

For each simulated trial, the control outcomes were generated from
\[
X_j=(X_{j1},X_{j2})^\top \sim N_2(0,\Sigma), \qquad j=1,\ldots,n_0,
\]
and the treatment outcomes were generated from
\[
Y_i=(Y_{i1},Y_{i2})^\top \sim N_2(\theta,\Sigma), \qquad i=1,\ldots,n_1,
\]
where
\[
\Sigma=
\begin{pmatrix}
1 & 0.30 \\
0.30 & 1
\end{pmatrix}.
\]
We used \(n_1=n_0=100\), and each setting was replicated 5000 times. 
A one-sided test at level 0.025 was used.

Five treatment-effect settings were considered:
\[
\theta=(0,0), \quad
\theta=(0.45,0), \quad
\theta=(0,0.45), \quad
\theta=(0.32,0.32), \quad
\theta=(0.45,-0.20).
\]
These correspond to the global null, an effect only on the higher-priority endpoint, an effect only on the lower-priority endpoint, effects on both endpoints, and effects in opposite directions across the two endpoints, respectively. 
The opposite-direction setting was included because, in practice, not all components of a composite endpoint need to move in the same direction.

For a treatment-control pair, define
\[
d_1=Y_{i1}-X_{j1}, \qquad d_2=Y_{i2}-X_{j2}.
\]
The component-level comparison was
\[
q_k(d_k)=\operatorname{sign}(d_k), \qquad k=1,2.
\]
Thus, the component comparison itself was not smoothed; the simulation focused on how information was passed from the higher-priority endpoint to the lower-priority endpoint.

Four methods were compared. 
The first was the strict hierarchical rule, under which the second endpoint was used only when the first endpoint was exactly tied. 
Because the endpoints were continuous, this rule was essentially determined by the first endpoint. 
The second was a hard-threshold rule,
\[
S_{\mathrm{HT},\tau}(Y_i,X_j)
=
\{1-I(|d_1|\le \tau)\}q_1(d_1)
+
I(|d_1|\le \tau)q_2(d_2).
\]
The third used a Gaussian weight,
\[
w_{\mathrm{G}}(d_1;\tau)
=
\exp\left\{-\frac{d_1^2}{2\tau^2}\right\},
\]
with score
\[
S_{\mathrm{G},\tau}(Y_i,X_j)
=
\{1-w_{\mathrm{G}}(d_1;\tau)\}q_1(d_1)
+
w_{\mathrm{G}}(d_1;\tau)q_2(d_2).
\]
The fourth used a logistic weight,
\[
w_{\mathrm{L}}(d_1;\tau,a)
=
\frac{1}{1+\exp\{(|d_1|-\tau)/a\}},
\]
where \(a=0.35\), and the corresponding score \(S_{\mathrm{L},\tau}\) was defined in the same way. 
The tuning parameter was varied over
\[
\tau \in \{0.05,0.10,0.20,0.35,0.50,0.75,1.00,1.50,2.00,3.00,5.00\}.
\]

For each method, the net-benefit parameter was estimated by
\[
\widehat{\Delta}_{\tau}
=
\frac{1}{n_1n_0}
\sum_{i=1}^{n_1}
\sum_{j=1}^{n_0}
S_{\tau}(Y_i,X_j).
\]
The standard error was estimated using the empirical first-order projection of the two-sample U-statistic.

Under the global null, the empirical rejection rates were close to the nominal level for all methods (Table~\ref{tab:sim1-type1}). 
Across the values of \(\tau\), the rejection rates ranged from 0.02500 to 0.03000 for the hard-threshold rule, from 0.02500 to 0.03000 for the Gaussian weight, and from 0.02480 to 0.03000 for the logistic weight. 
The strict hierarchy had an empirical rejection rate of 0.02480. 
The mean estimated net benefit was close to zero for all methods under the null.

\begin{table}[htbp]
\centering
\caption{Simulation study 1: empirical type I error under the global null.}
\label{tab:sim1-type1}
\begin{tabular}{lcccc}
\toprule
Method & Minimum & Median & Maximum & Nominal level \\
\midrule
Strict hierarchy & 0.02480 & 0.02480 & 0.02480 & 0.02500 \\
Hard threshold   & 0.02500 & 0.02620 & 0.03000 & 0.02500 \\
Gaussian weight  & 0.02500 & 0.02640 & 0.03000 & 0.02500 \\
Logistic weight  & 0.02480 & 0.02660 & 0.03000 & 0.02500 \\
\bottomrule
\end{tabular}
\end{table}

The results under the alternatives are summarized in Table~\ref{tab:sim1-best}. 
When the treatment effect was present on both endpoints, the strict hierarchy had power 0.60720. 
Allowing the second endpoint to contribute improved power substantially. 
The best rejection rates were 0.77980 for the hard-threshold rule, 0.78120 for the Gaussian weight, and 0.78380 for the logistic weight, all attained at \(\tau=1.00\). 
Thus, when both endpoints carried treatment information, a moderate value of \(\tau\) gave the best performance.

When the effect was present only on the higher-priority endpoint, the strict hierarchy was best, with power 0.86980. 
The hard-threshold and Gaussian rules with very small \(\tau\) performed similarly, with powers 0.86340 and 0.86220, respectively. 
Power decreased as \(\tau\) increased, reflecting dilution of the first endpoint by the second endpoint, which contained no treatment signal in this setting.

When the effect was present only on the lower-priority endpoint, the strict hierarchy had essentially no power, with an empirical rejection rate of 0.02440. 
In contrast, the hard-threshold, Gaussian, and logistic rules achieved powers 0.87640, 0.86920, and 0.87620, respectively, at \(\tau=5.00\). 
This setting illustrates the main motivation for relaxing the strict hierarchy: when the treatment effect is carried by a lower-priority endpoint, the strict rule may discard most of the relevant information.

In the opposite-direction setting, the strict hierarchy again had the highest rejection rate, 0.87460. 
The hard-threshold and Gaussian rules with very small \(\tau\) were close to the strict rule, with powers 0.86020 and 0.85540. 
However, increasing \(\tau\) caused a sharp decline in power and in the estimated net benefit, because the lower-priority endpoint favored control. 
This scenario shows that relaxing the hierarchy may be harmful when component effects are not aligned.

\begin{table}[htbp]
\centering
\caption{Simulation study 1: best empirical rejection rate over tuning parameters.}
\label{tab:sim1-best}
\begin{tabular}{llccc}
\toprule
Scenario & Method & Selected \(\tau\) & Rejection rate & Mean \(\widehat{\Delta}\) \\
\midrule
Both endpoints & Strict hierarchy & --   & 0.60720 & 0.17978 \\
Both endpoints & Hard threshold   & 1.00 & 0.77980 & 0.21145 \\
Both endpoints & Gaussian weight  & 1.00 & 0.78120 & 0.20438 \\
Both endpoints & Logistic weight  & 1.00 & 0.78380 & 0.20432 \\
\midrule
Higher-priority only & Strict hierarchy & --   & 0.86980 & 0.24859 \\
Higher-priority only & Hard threshold   & 0.05 & 0.86340 & 0.24629 \\
Higher-priority only & Gaussian weight  & 0.05 & 0.86220 & 0.24561 \\
Higher-priority only & Logistic weight  & 0.05 & 0.83640 & 0.22775 \\
\midrule
Lower-priority only & Strict hierarchy & --   & 0.02440 & 0.00090 \\
Lower-priority only & Hard threshold   & 5.00 & 0.87640 & 0.25091 \\
Lower-priority only & Gaussian weight  & 5.00 & 0.86920 & 0.24234 \\
Lower-priority only & Logistic weight  & 5.00 & 0.87620 & 0.25072 \\
\midrule
Opposite direction & Strict hierarchy & --   & 0.87460 & 0.24996 \\
Opposite direction & Hard threshold   & 0.05 & 0.86020 & 0.24453 \\
Opposite direction & Gaussian weight  & 0.05 & 0.85540 & 0.24306 \\
Opposite direction & Logistic weight  & 0.05 & 0.78380 & 0.21352 \\
\bottomrule
\end{tabular}
\end{table}

Figure~\ref{fig:sim1-power} shows the rejection rate as a function of \(\tau\). 
The shape of the curves was consistent with the endpoint-level treatment effects. 
In the higher-priority-only setting, the curves decreased as \(\tau\) increased. 
In the lower-priority-only setting, the curves increased with \(\tau\). 
In the setting with effects on both endpoints, the curves peaked at an intermediate value of \(\tau\). 
In the opposite-direction setting, increasing \(\tau\) reduced the treatment signal and eventually led to negative estimated net benefit.

The Gaussian and logistic weights gave similar qualitative conclusions, but their behavior differed in some regions of \(\tau\). 
For the lower-priority-only setting, the Gaussian weight had higher power than the hard-threshold rule at moderate values of \(\tau\); for example, at \(\tau=1.00\), the power was 0.51540 for the Gaussian weight and 0.41400 for the hard-threshold rule. 
At large \(\tau\), the hard-threshold and logistic rules were slightly better, because they more completely passed the comparison to the lower-priority endpoint. 
For the setting with effects on both endpoints, all relaxed rules were close at their best tuning values. 
The logistic rule was generally closer to the hard-threshold rule than the Gaussian rule, which is expected from the shape of the logistic transition with \(a=0.35\).

Overall, this first simulation suggests that relaxing a strict hierarchy can recover substantial power when the treatment effect is mainly expressed through a lower-priority endpoint. 
At the same time, relaxing the hierarchy can reduce power when the higher-priority endpoint is the main or only favorable component, and it may be particularly undesirable when the lower-priority endpoint moves in the opposite direction. 
Thus, the tuning parameter should not be viewed as a purely numerical choice; it determines how strongly the analysis preserves the priority of the higher-level endpoint.

\begin{figure}[htbp]
\centering
\includegraphics[width=0.95\textwidth]{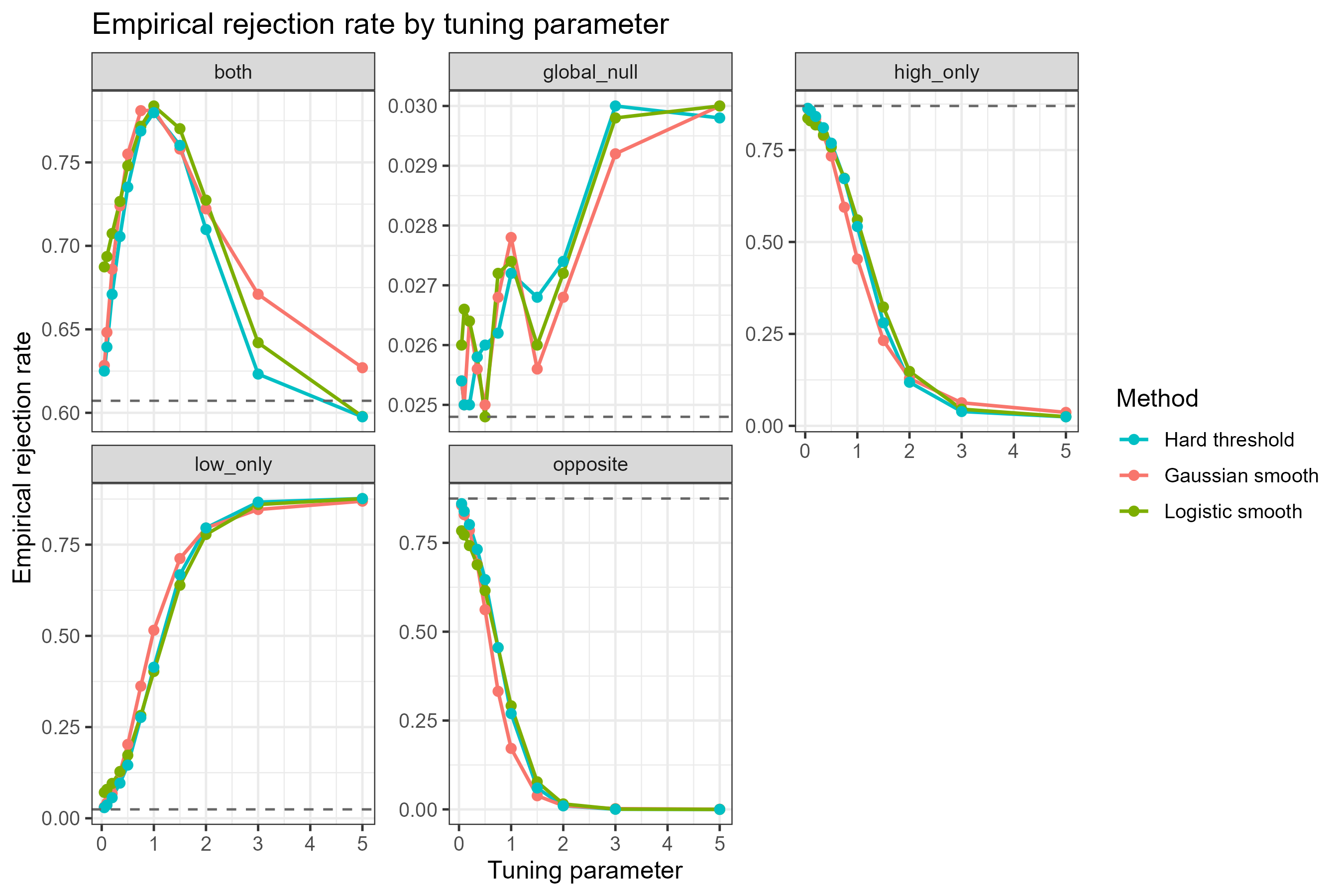}
\caption{Simulation study 1: empirical rejection rate by tuning parameter. The dashed horizontal line represents the strict hierarchical rule.}
\label{fig:sim1-power}
\end{figure}

\begin{figure}[htbp]
\centering
\includegraphics[width=0.95\textwidth]{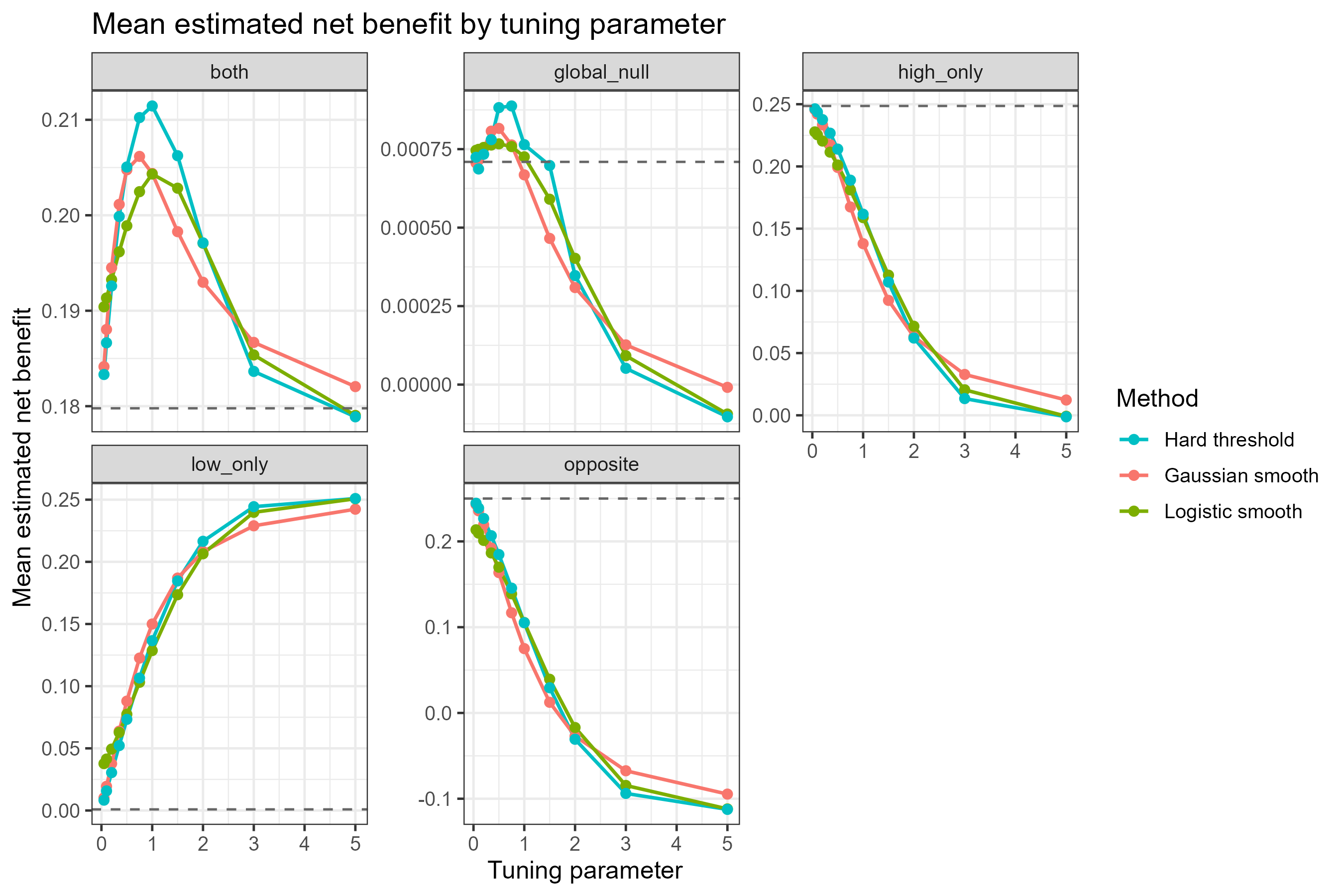}
\caption{Simulation study 1: mean estimated net benefit by tuning parameter. The dashed horizontal line represents the strict hierarchical rule.}
\label{fig:sim1-delta}
\end{figure}

\section{Discussion}

We proposed a calibrated soft-loading pairwise comparison framework for hierarchical composite endpoints. 
The method retains the pairwise comparison structure of the win ratio and generalized pairwise comparisons, while replacing hard transitions between endpoint levels with continuous loading functions. 
The resulting estimator is a two-sample U-statistic, permitting standard large-sample inference through Hoeffding projection.

The main motivation is not that soft loading uniformly improves upon the win ratio. 
Rather, the framework makes explicit a trade-off that is already present in hierarchical composite endpoint analysis. 
Strict hierarchy protects the clinical priority of high-ranking outcomes, but it may underuse lower-priority outcomes when those outcomes carry most of the treatment signal. 
Soft loading allows lower-priority outcomes to contribute when higher-priority outcomes are close to tied, at the cost of potentially diluting high-priority signals.

This perspective also clarifies the relationship between the proposed method and threshold-based extensions of the win ratio. 
Hard-threshold methods relax the hierarchy by using discontinuous decision boundaries. 
The proposed method uses continuous loading functions, allowing smooth tuning and efficiency-based parameter selection. 
This may be useful when the analyst wants to avoid abrupt changes in pairwise decisions induced by small changes around a threshold.

Several extensions remain. 
First, time-to-event endpoints with censoring require additional treatment, such as inverse probability weighting or probabilistic pairwise scores. 
Second, the effect of endpoint correlation deserves careful investigation, particularly because conditioning on near-ties at high-priority levels may change the distribution of lower-priority endpoint differences. 
Third, data-adaptive selection of the loading parameter should account for type I error inflation; sample splitting or pre-specification may be needed for confirmatory analyses.

In summary, the calibrated soft-loading framework provides a continuous and interpretable extension of hierarchical pairwise comparisons. 
It offers a principled way to study how much lower-priority endpoint information should be allowed to contribute while preserving the clinical structure of prioritized composite endpoints.

\section*{Disclosure Statement}

The author spends the majority of time supporting clinical trials and has limited bandwidth for research; any remaining imperfections are acknowledged with apology.

\bibliographystyle{plainnat}
\bibliography{references}

\end{document}